\documentclass{amsart}

\usepackage{amsmath,amsthm,amssymb,bm}
\usepackage{MnSymbol} 
\usepackage[backref]{hyperref}
\usepackage{a4wide}
\usepackage{cleveref}
\usepackage{graphicx,color}
\usepackage{tikz}
\usepackage{ytableau}
\numberwithin{equation}{section}
\usepackage{caption}
\usepackage{subcaption}
\newtheorem{thm}{Theorem}[section]   
\newtheorem{lem}[thm]{Lemma}
\newtheorem{prop}[thm]{Proposition}
\newtheorem{cor}[thm]{Corollary}
\theoremstyle{definition}

\newtheorem{defn}[thm]{Definition}

\newcommand\SE{\operatorname{SE}}

\newcommand\QQ{\mathbb{Q}}

\newcommand{\ZZ}{\mathbb{Z}}

\newcommand\LL{\mathcal{L}}

\newcommand\Mot{\operatorname{Mot}}

\newcommand\TRIL[1]{
  \def\X{#1}
  \foreach \i in {0,...,\X}
  {
    \draw[gray,very thin] (\i,0) -- (\i,\i);
    \draw[gray,very thin] (\i,\i) -- (\X,\i);
  }
}

\newcommand\TRIDL[1]{
\def\X{#1}
    \foreach \i in {1,...,\X}
  {
    \draw[gray,very thin] (\i,\X) -- (\i,\X-\i);
    \draw[gray,very thin] (\i,\X -\i) -- (\X,\X-\i);
  }
  }

\newcommand\LHLL[2]{
  \def\X{#1} \def\Y{#2}
  \foreach \i in {0,...,\X}
  {
\pgfmathsetmacro{\m}{\Y};
    \draw[gray,very thin] (\i,0) -- (\i,\m);
  }
\pgfmathsetmacro{\m}{\Y-1};
  \foreach \x in {1,...,\X}
{ \foreach \j in {0,...,\m}
  {
\draw[gray,very thin] (\x-1,\j) -- (\x,\j);
}}
\pgfmathsetmacro{\m}{\Y-1};
 \foreach \i in {2,...,\X}
 \foreach \j in {0,...,\m} 
{\foreach \y in {2,...,\i}
{
\pgfmathsetmacro{\w}{\j+(\y-1)/\i};
\pgfmathsetmacro{\z}{\j+(\y-1)/(\i+1)};
\draw[gray,very thin] (\i-1,\w) -- (\i,\z);
}}
}

\newcommand\DLHLL[2]{
  \def\X{#1} \def\Y{#2}
  \foreach \i in {0,...,\X}
  {
\pgfmathsetmacro{\m}{\Y};
    \draw[gray,very thin] (\i,0) -- (\i,\m);
  }
\pgfmathsetmacro{\m}{\Y};
  \foreach \x in {1,...,\X}
{ \foreach \j in {0,...,\m}
  {
}}
\pgfmathsetmacro{\m}{\Y-1};
 \foreach \i in {1,...,\X}
 \foreach \j in {0,...,\m} 
{\foreach \y in {1,...,\i}
{
\pgfmathsetmacro{\w}{\j+(\y-1)/\i};
\pgfmathsetmacro{\z}{\j+(\y)/(\i+1)};
\draw[gray,very thin] (\i-1,\w) -- (\i,\z);
}}
}

\newcommand\SLHlabel[1]{
\def\X{#1}
  \foreach \i in {0,...,\X}
  {
    \node at (\i+0.05,.9) {\i};
  }
  \node at (-.3,0) {0};
  \node at (-.3,1) {1};
  \node at (-.3,2) {2};
}

\newcommand\SLH[1]{
\def\X{#1}
\foreach \i in {0,...,\X}
  {
  \pgfmathsetmacro{\m}{2};
  \draw[gray,very thin] (\i,0) -- (\i,\m);
  }
  \draw[gray,very thin] (0,1) -- (\X,1);

\pgfmathsetmacro{\m}{1};
\foreach \i in {2,...,\X}
\foreach \j in {1,...,\m} 
{\foreach \y in {2,...,\i}
{
\pgfmathsetmacro{\w}{\j+(\y-1)/\i};
\pgfmathsetmacro{\z}{\j+(\y-1)/(\i+1)};
\draw[gray,very thin] (\i-1,\w) -- (\i,\z);
}}

\pgfmathsetmacro{\m}{0}; 
\foreach \i in {1,...,\X}
\foreach \j in {0,...,\m}
{\foreach \y in {1,...,\i}
   {
   \pgfmathsetmacro{\w}{\j+(\y-1)/\i};
   \pgfmathsetmacro{\z}{\j+(\y)/(\i+1)};
   \draw[gray,very thin] (\i-1,\w) -- (\i,\z);
}}
}

\title{Quantum interpretation of lattice paths}
 
 \author{Bhargavi Jonnadula}
 \address{Mathematical Institute, University of Oxford, Oxford OX2 6GG, UK}
 \email{bhargavi.jonnadula@maths.ox.ac.uk}

  \author{Jonathan P. Keating}
 \address{Mathematical Institute, University of Oxford, Oxford OX2 6GG, UK}
 \email{keating@maths.ox.ac.uk}

\begin{document}

\begin{abstract}
In the 1980s, Viennot \cite{ViennotLN,ViennotOP} developed a combinatorial approach to studying mixed moments of orthogonal polynomials using Motzkin paths.  Recently, an alternative combinatorial model for these mixed moments based on lecture hall paths was introduced in \cite{Corteel2023}. For sequences of orthogonal polynomials, we establish here a bijection between the Motzin paths and the lecture hall paths via a novel symmetric lecture hall graph.  We use this connection to calculate the moments of the position operator in various separable quantum systems, such as the quantum harmonic oscillator and the hydrogen atom, showing that they may be expressed as generating functions of Motzkin paths and symmetric lecture hall paths, thereby providing a quantum interpretation for these paths.  Our approach can be extended to other quantum systems where the wavefunctions are expressed in terms of orthogonal polynomials.
\end{abstract}

\maketitle
\tableofcontents


\section{Introduction}
Orthogonal polynomials play a central role in many areas of mathematics and physics. Since the 1980s, the groundbreaking work of Viennot \cite{ViennotLN, ViennotOP} and Flajolet \cite{Flajolet1980} has led to several interesting combinatorial interpretations of orthogonal polynomials, thereby providing novel combinatorial framings of the problems in which they arise. 

Our focus here will be on lattice paths that arise in the combinatorial representations of orthogonal polynomials. We study two types of lattice paths: (1) the Motzkin paths introduced by Viennot and (2) paths on the lecture hall graph introduced by Corteel and Kim \cite{CorteelLHT}. The lecture hall paths were first studied in the context of little $q$-Jacobi polynomials, but have recently been extended to the whole Askey scheme in \cite{Corteel2023}. As one of our main results, we provide a bijection between these two paths, by introducing a new graph that we call the symmetric lecture hall graph, see \Cref{sec:lhg}. 

As our second result, we give a quantum mechanical interpretation of both Motzkin and symmetric lecture paths.  In order to give concrete examples, we focus on the Hermite and Laguerre polynomials that arise in the quantum harmonic oscillator and the hydrogen atom respectively. These are canonical examples of separable systems in which the quantum wavefunctions can be calculated exactly and are related to orthogonal polynomials. The quantum harmonic oscillator is of fundamental importance in many areas of physics, such as quantum field theory, lattice vibrations and phonons, and electrodynamics. The hydrogen atom is special because, being a two-body problem, it is one of the few atomic systems that can be solved analytically. In both cases, we can express the moments of the position operator as weighted paths on the lattice graphs. 

This paper is organised as follows.  In \Cref{sec:mot_paths}, we recall Viennot's theory for Motzkin paths. In \Cref{sec:lhg}, we define symmetric lecture hall graphs and provide a connection with Motzkin paths. In \Cref{sec:qho}, we detail the connection between moments of the position operator of the quantum harmonic oscillator and both Motzkin and symmetric lecture hall paths. Finally, in \Cref{sec:hatom} we extend these results to the case of the hydrogen atom.

\section{Motzkin Paths}
\label{sec:mot_paths}
Orthogonal polynomials can be defined as a sequence $\{p_n(x)\}_{n\ge0}$ of polynomials with $ \deg p_n(x) =n$  such that there is a linear functional $\LL$  satisfying
 $\LL(p_n(x)p_m(x)) = K_n \delta_{n,m} $, where $K_n\ne 0$ . Their 
\emph{moments $\{\sigma_n\}_{n\ge0} $} are defined by
\begin{equation*}
    \sigma_n = \frac{\LL(x^n)}{\LL(1)}. 
\end{equation*}
More generally, the \emph{mixed moments $\{\sigma_{n,k}\}_{n,k\ge 0}$} are defined by
\begin{equation*}
\sigma_{n,k} = \frac{\LL(x^np_k(x))}{\LL(p_k(x)^2)}.    
\end{equation*}
By the orthogonality,
\begin{equation}\label{eq:def:sigma}
  x^n = \sum_{k=0}^{n} \sigma_{n,k} p_k(x),
\end{equation}
which can also be taken as the definition of the mixed moments. 

Viennot \cite{ViennotLN, ViennotOP} found a combinatorial interpretation for $\sigma_{n,k}$ using Motzkin paths. To illustrate,
suppose that $\{ p_n(x) \}_{n\ge 0}$ is a sequence of monic orthogonal polynomials satisfying the 3-term recurrence
\begin{equation}
\label{eq:3-term}
p_{n+1}(x) = (x-b_n)p_n(x)-\lambda_{n}p_{n-1}(x),\quad \lambda_0=0.
\end{equation}
A \emph{Motzkin path} $p$ of length $|p|=c-a$ is a path on $\mathbb{Z}^2$ from $(a, b)$ to $(c,d)$ consisting of up steps $(1,1) $, horizontal steps $(1,0)$, and downs steps $(1,-1)$ that never go below the line $y = 0$. The weight of a Motzkin path $w_{\Mot}(p)$ is the product of the weights of the steps in $p $, where the weight of an up step is always $1 $, the weight of a horizontal step at height $k$ is $b_k$, and the weight of a down step starting at height $k$ is $\lambda_k$. Viennot showed that $\sigma_{n,k}$ is the sum of $ w(p)$ for all Motzkin paths $ p$ from $(0,0)$ to $(n,k)$. For example,
\begin{equation*}
   \sigma_{3,1} = b_0^2+b_0b_1+b_1^2+\lambda_1+\lambda_2   
\end{equation*}
is the generating function for all Motzkin paths from $(0,0)$ to $(3,1)$  as shown in \Cref{fig:motzkin}. Using this combinatorial
interpretation for $\sigma_{n,k}$ one can show that the moments of Hermite, Charlier, and Laguerre polynomials are generating functions
for perfect matchings, set partitions, and permutations, respectively.
See \cite{CKS, ViennotLN, ViennotOP, Zeng2021} for more details on the combinatorics of orthogonal polynomials.

\begin{figure}
  \hspace*{\fill}
    \begin{subfigure}[b]{0.19\textwidth}
\centering
\begin{tikzpicture}[scale = 0.65]
\draw(-0.25,-0.25) grid (3.5,2.5);
\foreach \x in {0,1,...,3} { \node [anchor=north] at (\x,-0.3) {\x}; }
\foreach \y in {0,1,2} { \node [anchor=east] at (-0.3,\y) {\y}; }
\draw [thick] (0,0) -- (1,0) -- (2,0)-- (3,1);
\filldraw[black](0,0)circle[radius=2pt];
\filldraw[black](1,0)circle[radius=2pt];
\filldraw[black](2,0)circle[radius=2pt];
\filldraw[black](3,1)circle[radius=2pt];
\node [anchor=north] at (0.5,0.75) {$b_0$};
\node [anchor=north] at (1.5,0.75) {$b_0$};
\end{tikzpicture}
\captionsetup{labelformat=empty}
\caption{$w_{\Mot}(p) =b_0^2$}
\end{subfigure} 
   \begin{subfigure}[b]{0.19\textwidth}
\centering
\begin{tikzpicture}[scale = 0.65]
\draw(-0.25,-0.25) grid (3.5,2.5);
\foreach \x in {0,1,...,3} { \node [anchor=north] at (\x,-0.3) {\x}; }
\draw [thick] (0,0) -- (1,0) -- (2,1)-- (3,1);
\filldraw[black](0,0)circle[radius=2pt];
\filldraw[black](1,0)circle[radius=2pt];
\filldraw[black](2,1)circle[radius=2pt];
\filldraw[black](3,1)circle[radius=2pt];
\node [anchor=north] at (0.5,0.75) {$b_0$};
\node [anchor=north] at (2.5,1.75) {$b_1$};
\end{tikzpicture}
\captionsetup{labelformat=empty}
\caption{$w_{\Mot}(p) =b_0 b_1$}
\end{subfigure} 
  \begin{subfigure}[b]{0.19\textwidth}
\centering
\begin{tikzpicture}[scale = 0.65]
\draw(-0.25,-0.25) grid (3.5,2.5);
\foreach \x in {0,1,...,3} { \node [anchor=north] at (\x,-0.3) {\x}; }
\draw [thick] (0,0) -- (1,1) -- (2,1)-- (3,1);
\filldraw[black](0,0)circle[radius=2pt];
\filldraw[black](1,1)circle[radius=2pt];
\filldraw[black](2,1)circle[radius=2pt];
\filldraw[black](3,1)circle[radius=2pt];
\node [anchor=north] at (1.5,1.75) {$b_1$};
\node [anchor=north] at (2.5,1.75) {$b_1$};
\end{tikzpicture}
\captionsetup{labelformat=empty}
\caption{$w_{\Mot}(p) = b_1^2$}
\end{subfigure} 
\begin{subfigure}[b]{0.19\textwidth}
\centering
\begin{tikzpicture}[scale = 0.65]
\draw(-0.25,-0.25) grid (3.5,2.5);
\foreach \x in {0,1,...,3} { \node [anchor=north] at (\x,-0.3) {\x}; }
\draw [thick] (0,0) -- (1,1) -- (2,0)-- (3,1);
\filldraw[black](0,0)circle[radius=2pt];
\filldraw[black](1,1)circle[radius=2pt];
\filldraw[black](2,0)circle[radius=2pt];
\filldraw[black](3,1)circle[radius=2pt];
\node [anchor=north] at (1.75,1.1) {$\lambda_1$};
\end{tikzpicture}
\captionsetup{labelformat=empty}
\caption{$w_{\Mot}(p) = \lambda_1$}
\end{subfigure} 
\begin{subfigure}[b]{0.19\textwidth}
\centering
\begin{tikzpicture}[scale = 0.65]
\draw(-0.25,-0.25) grid (3.5,2.5);
\foreach \x in {0,1,...,3} { \node [anchor=north] at (\x,-0.3) {\x}; }
\draw [thick] (0,0) -- (2,2) -- (3,1);
\filldraw[black](0,0)circle[radius=2pt];
\filldraw[black](1,1)circle[radius=2pt];
\filldraw[black](2,2)circle[radius=2pt];
\filldraw[black](3,1)circle[radius=2pt];
\node [anchor=north] at (2.75,2.1) {$\lambda_2$};
\end{tikzpicture}
\captionsetup{labelformat=empty}
\caption{$w_{\Mot}(p) = \lambda_2$}
\end{subfigure}\hspace{0.0cm} 
  \hspace*{\fill}
  \caption{ The Motzkin paths from \( (0,0) \) to \( (3,1) \) and
    their weights. The weights of horizontal and down steps are
    indicated.}
\label{fig:motzkin}
\end{figure}
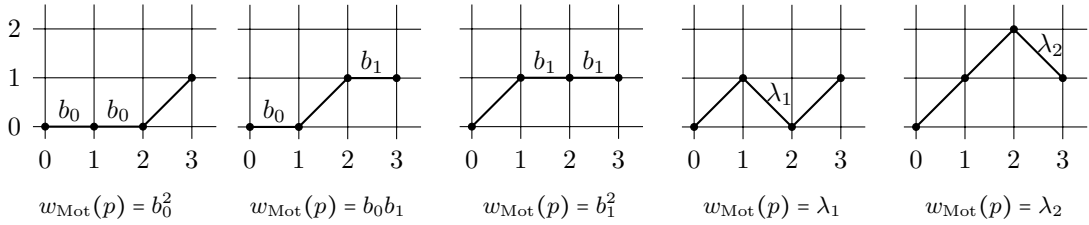

The mixed moments can be generalised by considering the linearisation expansion of $x^np_m(x)$ in the orthogonal polynomial basis as
\begin{equation*}
    x^n p_m(x) = \sum_{k=0}^{n+m}\sigma_{n,m,k}p_k(x),
\end{equation*}
where 
\begin{equation}
\label{eq:lin_coef_def}
\sigma_{n,m,k} = \frac{\LL(x^n p_m(x) p_k(x))}{\LL(p^2_k)}.
\end{equation}
Instead of the monomial, one can consider any polynomial in $x$, say $f(x)$, and study the expansion of $f(x)p_m(x)$ in the basis $\{p_k(x)\}_{k\geq 0}$.

\begin{thm}[\cite{ViennotLN, ViennotOP}]
\label{thm:lin_coef}
    For $n,m,k\geq 0$, consider a set of Motzkin paths $p$ of length $|p|=n$. Let $p:m\leadsto k$ be the set of paths starting from height $m$ and ending at height $k$. Then the coefficients $\sigma_{n,m,k}$ are generating functions given by
    \begin{equation}
        \label{eq:lin_coef}
        \sigma_{n,m,k} =\sum_{p:m\leadsto k,\,  |p|=n} w_{\Mot}(p).
    \end{equation}
\end{thm}
Note that when $n=0$, $\sigma_{0,m,k}=\delta_{km}$ recovering the orthogonality relation.

\section{Lecture hall graphs}
\label{sec:lhg}
In this section, we give basic definitions and lemmas on lecture hall graphs that we use later to establish a connection to the moments of the quantum harmonic oscillator and the hydrogen atom.

\subsection{Symmetric lecture hall graphs} Corteel and Kim introduced the lecture hall graph and its dual in their study of little $q$-Jacobi polynomials \cite{CorteelLHT}. The theory was subsequently extended to the entire $q$-Askey and the Askey scheme in \cite{Corteel2023}. Some properties of the graph and further applications can be found in \cite{Corteel2019enum, Corteel2021arc}. Here we consider a symmetric lecture hall graph, which is a fusion of the lecture hall graph and the dual lecture hall graph.

For nonnegative integers $i,j$ with $j\le i$ and $t\in \{0,1\}$,
we denote
\begin{equation}
  v^t_{i,j} = 
      \left(i,t+ \frac{j}{i+1}\right) \in \ZZ\times\QQ. 
\end{equation}

\begin{defn}
  The \emph{symmetric lecture hall graph} is the (undirected) graph $\mathcal{G} = (V,E)$, where
\begin{align*}
  V &= \{v^t_{i,j}: i,j\in \ZZ_{\ge0}, 0\le j\le i, 0\le t\le 1 \},\\
  E &= \{(v^t_{i,j}, v^t_{i,j+1})\in V^2 : i,j\in\ZZ_{\ge0},0\le t\le 1  \} 
        \cup \{(v^t_{i,j}, v^t_{i+1,j})\in V^2 : i,j\in\ZZ_{\ge0},0\le t\le 1  \}.
\end{align*}
Here, the notation $\{(v^t_{i,j}, v^t_{i,j+1})\in V^2 : i,j\in\ZZ_{\ge0},0\le t\le 1  \}$ means  the set of all elements in $V^2$ of the form $(v^t_{i,j}, v^t_{i,j+1})$ for some $i,j\in\ZZ_{\ge0}$ and $0\le t\le 1 $. See
\Cref{fig:slhg_def}. 
\end{defn}

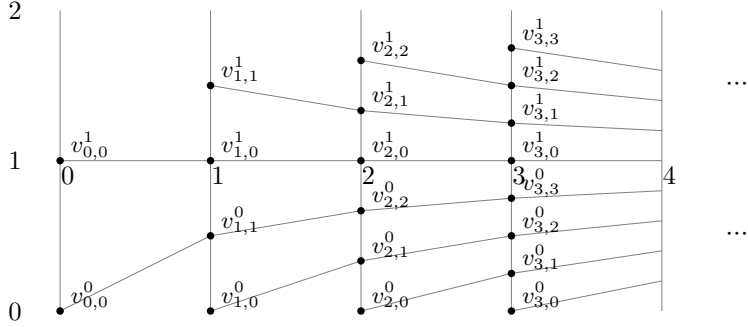
\begin{figure}
  \centering
  \begin{tikzpicture}[scale=2]
    \SLHlabel{4}
    \SLH{4}
        \begin{scope}
      \node at (0,0) [circle,fill,inner sep=1pt]{};
      \node at (1,0) [circle,fill,inner sep=1pt]{};
      \node at (1,1/2) [circle,fill,inner sep=1pt]{};
      \node at (2,0) [circle,fill,inner sep=1pt]{};
      \node at (2,1/3) [circle,fill,inner sep=1pt]{};
      \node at (2,2/3) [circle,fill,inner sep=1pt]{};
      \node at (3,0) [circle,fill,inner sep=1pt]{};
      \node at (3,1/4) [circle,fill,inner sep=1pt]{};
      \node at (3,2/4) [circle,fill,inner sep=1pt]{};
      \node at (3,3/4) [circle,fill,inner sep=1pt]{};
      \node at (0,1) [circle,fill,inner sep=1pt]{};
      \node at (1,1) [circle,fill,inner sep=1pt]{};
      \node at (1,3/2) [circle,fill,inner sep=1pt]{};
      \node at (2,1) [circle,fill,inner sep=1pt]{};
      \node at (2,4/3) [circle,fill,inner sep=1pt]{};
      \node at (2,5/3) [circle,fill,inner sep=1pt]{};
      \node at (3,1) [circle,fill,inner sep=1pt]{};
      \node at (3,5/4) [circle,fill,inner sep=1pt]{};
      \node at (3,6/4) [circle,fill,inner sep=1pt]{};
      \node at (3,7/4) [circle,fill,inner sep=1pt]{};
    \end{scope}
    \begin{scope}[shift={(0.2,0.1)}]
      \small
      \node at (0,0) {\( v^0_{0,0} \)};
      \node at (1,0) {\( v^0_{1,0} \)};
      \node at (1,1/2) {\( v^0_{1,1} \)};
      \node at (2,0) {\( v^0_{2,0} \)};
      \node at (2,1/3) {\( v^0_{2,1} \)};
      \node at (2,2/3) {\( v^0_{2,2} \)};
      \node at (3,0) {\( v^0_{3,0} \)};
      \node at (3,1/4) {\( v^0_{3,1} \)};
      \node at (3,2/4) {\( v^0_{3,2} \)};
      \node at (3,3/4) {\( v^0_{3,3} \)};
      \node at (0,1) {\( v^1_{0,0} \)};
      \node at (1,1) {\( v^1_{1,0} \)};
      \node at (1,3/2) {\( v^1_{1,1} \)};
      \node at (2,1) {\( v^1_{2,0} \)};
      \node at (2,4/3) {\( v^1_{2,1} \)};
      \node at (2,5/3) {\( v^1_{2,2} \)};
      \node at (3,1) {\( v^1_{3,0} \)};
      \node at (3,5/4) {\( v^1_{3,1} \)};
      \node at (3,6/4) {\( v^1_{3,2} \)};
      \node at (3,7/4) {\( v^1_{3,3} \)};
    \end{scope}
    \node at (4.5,0.5) {$\cdots$ };
    \node at (4.5,1.5) {$\cdots $};
  \end{tikzpicture}
  \caption{Symmetric lecture hall graph with vertices indicated.}
  \label{fig:slhg_def}
  \end{figure}

We consider two kinds of steps:
\begin{itemize}
\item an \emph{east step} is a directed edge of the form
   $(v^t_{i,j}, v^t_{i+1,j})$ in $\mathcal{G}$,
\item a \emph{south step} is a directed edge of the form
  $(v^t_{i,j}, v^t_{i,j-1})$ or $(v^t_{i,0}, v^{t-1}_{i,i})$.
\end{itemize}

For two vertices $u$ and $v$ in $\mathcal{G}$ , a \emph{path} from $u$ to $ v$ is a sequence $(v_0,v_1,\dots,v_n)$ of vertices of $\mathcal{G}$ such that $v_0=u$, $v_n=v$, and each $(v_{i}, v_{i+1})$ is an east or south step. Let $\SE(u\to v)$  denote the set of paths from $u$ to $v$ in $\mathcal{G}$ consisting of south and east steps. 

A \emph{weight system} is a function $w $ that assigns a weight $w(s)$  to each east step $s=(v^t_{i,j}, v^t_{i+1,j})$ in $\mathcal{G}$. We denote
\begin{equation*}
   w(t;i,j) := w(v^t_{i,j},v^t_{i+1,j}).   
\end{equation*}
In other words, $ w(t;i,j)$ is the weight of the $j$th east step from the bottom, where the bottommost one is the $0$th step, among the east steps in the region $\{(x,y): i\le x\le i+1, t-1\le y<t \}$, see \Cref{fig:slhg_wt}. Given a weight system $w$, the weight $ w(p)$ of a path $p$ is defined to be the product of the weights of all east steps in $p$.

\begin{figure}
  \centering
  \begin{tikzpicture}[scale=2]
    \SLHlabel{4}
    \SLH{4}
    \begin{scope}
      \node at (0,0) [circle,fill,inner sep=1pt]{};
      \node at (1,0) [circle,fill,inner sep=1pt]{};
      \node at (1,1/2) [circle,fill,inner sep=1pt]{};
      \node at (2,0) [circle,fill,inner sep=1pt]{};
      \node at (2,1/3) [circle,fill,inner sep=1pt]{};
      \node at (2,2/3) [circle,fill,inner sep=1pt]{};
      \node at (3,0) [circle,fill,inner sep=1pt]{};
      \node at (3,1/4) [circle,fill,inner sep=1pt]{};
      \node at (3,2/4) [circle,fill,inner sep=1pt]{};
      \node at (3,3/4) [circle,fill,inner sep=1pt]{};
      \node at (0,1) [circle,fill,inner sep=1pt]{};
      \node at (1,1) [circle,fill,inner sep=1pt]{};
      \node at (1,3/2) [circle,fill,inner sep=1pt]{};
      \node at (2,1) [circle,fill,inner sep=1pt]{};
      \node at (2,4/3) [circle,fill,inner sep=1pt]{};
      \node at (2,5/3) [circle,fill,inner sep=1pt]{};
      \node at (3,1) [circle,fill,inner sep=1pt]{};
      \node at (3,5/4) [circle,fill,inner sep=1pt]{};
      \node at (3,6/4) [circle,fill,inner sep=1pt]{};
      \node at (3,7/4) [circle,fill,inner sep=1pt]{};
    \end{scope}
    \begin{scope}[shift={(0.5,0.1)}]
      \small
      \node at (0,0.3) {\( w(0;0,0) \)};
      \node at (1,0.2) {\( w(0;1,0) \)};
      \node at (1,0.6) {\( w(0;1,1) \)};
      \node at (2,0) {\( w(0;2,0) \)};
      \node at (2,1/3) {\( w(0;2,1) \)};
      \node at (2,2/3) {\( w(0;2,2) \)};
      \node at (3,0) {\( w(0;3,0) \)};
      \node at (3,1/4) {\( w(0;3,1) \)};
      \node at (3,2/4) {\( w(0;3,2) \)};
      \node at (3,3/4) {\( w(0;3,3) \)};
      \node at (0,1) {\( w(1;0,0) \)};
      \node at (1,1) {\( w(1;1,0) \)};
      \node at (1,3/2) {\( w(1;1,1) \)};
      \node at (2,1) {\( w(1;2,0) \)};
      \node at (2,4/3) {\( w(1;2,1) \)};
      \node at (2,5/3) {\( w(1;2,2) \)};
      \node at (3,1) {\( w(1;3,0) \)};
      \node at (3,5/4) {\( w(1;3,1) \)};
      \node at (3,6/4) {\( w(1;3,2) \)};
      \node at (3,7/4) {\( w(1;3,3) \)};
    \end{scope}
    \draw [red, very thick] (4,0/5) -- (4,3/5) -- (3,2/4) -- (2,1/3) -- (2,4/3) -- (1,3/2) -- (1,4/2);
    \node at (4.5,0.5) {$\cdots$};
    \node at (4.5,1.5) {$\cdots$};
    \end{tikzpicture}
      \caption{Weights in symmetric lecture hall graph with a path in $\SE((1,2)\to (4,0))$.}
      \label{fig:slhg_wt}
\end{figure}
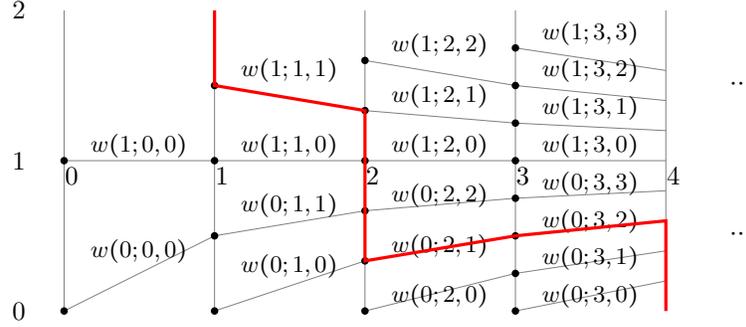

The $t=0$ and $t=1$ levels in the symmetric graph are the dual lecture hall graph and the lecture hall graph of height 1, respectively. Observe that these two levels can be identified with the staircase grid as shown in \Cref{fig:lhlg_grid,fig:dlhg_grid}. This bijection implies the following results.

\begin{prop}
    In the symmetric lecture hall graph, the total number of paths $p:(k,2)\to(n,1)$ and $p:(k,1)\to (n,0)$ are $\binom{n}{k}$ and $\frac{k+1}{n+1}\binom{2n-k}{n-k}$, respectively.
\end{prop}
\begin{proof}
    The enumeration problem in the case $p:(k,2)\to (n,1)$ is straightforward due to the bijection to the staircase grid. To count the number of paths $p$ such that $p:(k,1)\to(n,0)$, note that there are a total of $2n-k$ steps (excluding the first south step) out of which there are $n$ south steps and $n-k$ east steps. The total number of such paths is $\binom{2n-k}{n-k}$. Furthermore, we have the constraint that the number of south steps is strictly greater than the number of east steps. This is easily seen from the staircase grid in \Cref{fig:dlhg_grid} as the paths lie strictly above the diagonal. The proportion of paths that satisfy this constraint is $\frac{k+1}{n+1}$ by the Ballot theorem.  Therefore, the total number of paths $p:(k,1)\to (n,0)$ is $\frac{k+1}{n+1}\binom{2n-k}{n-k}$.
\end{proof}
\begin{figure}
\begin{subfigure}[b]{\textwidth}
\centering
  \begin{tikzpicture}[scale=1.5]
    \LHLL{5}1
    \draw[white, thick] (0,1) -- (5,1);
    \draw [red, very thick] (5,0) -- (5,1/6) -- (4,1/5)
     -- (4,2/5) -- (3,2/4) -- (2,2/3) -- (2,1);
     \node at (-0.3,0) {1};
     \node at (-0.3,1) {2};
     \node at (0,-0.2) {0};
     \node at (1,-0.2) {1};
     \node at (2,-0.2) {2};
     \node at (3,-0.2) {3};
     \node at (4,-0.2) {4};
     \node at (5,-0.2) {5};
     \node at (5.5,.5) {\(\Leftrightarrow\)};
   \end{tikzpicture}
  \begin{tikzpicture}[scale=0.5]
    \TRIL{5}
    \pgfmathsetmacro{\m}{5};
    \foreach \i in {0,...,\m}
    {
    \node at (\i,-0.3) {\i};
    }
    \draw [red, very thick] (5,0) -- (5,1) -- (4,1)
     -- (4,2) -- (3,2) -- (3,2) -- (2,2);
  \end{tikzpicture}
  \caption{A path in the lecture hall graph of height 1 and its corresponding path in the staircase lattice.}
  \label{fig:lhlg_grid}
  \end{subfigure}

  \begin{subfigure}[b]{\textwidth}
  \centering
  \begin{tikzpicture}[scale=1.5]
    \DLHLL{5}1
    \draw[white, thick] (0,1) -- (5,1);
    \draw [red, very thick] (2,1) -- (2,1/3) -- (3,2/4) -- (4,3/5) -- (4,1/5) -- (5,2/6) -- (5,0);
     \node at (-0.3,0) {0};
     \node at (-0.3,1) {1};
     \node at (0,-0.2) {0};
     \node at (1,-0.2) {1};
     \node at (2,-0.2) {2};
     \node at (3,-0.2) {3};
     \node at (4,-0.2) {4};
     \node at (5,-0.2) {5};
     \node at (5.5,.5) {\(\Leftrightarrow\)};
   \end{tikzpicture}
  \begin{tikzpicture}[scale=0.5]
\TRIDL{6}
\pgfmathsetmacro{\m}{5};
\foreach \i in {0,...,\m}
{
\node at (\i+1,-.3) {\i};
}
  \draw [red, very thick] (3,6) -- (3,4) -- (5,4) -- (5,2) -- (6,2) -- (6,0);    
  \end{tikzpicture}
  \caption{A path in the dual lecture hall graph of height 1 and its corresponding path in the staircase lattice.}
  \label{fig:dlhg_grid}
  \end{subfigure}
\end{figure}

\begin{lem}
\label{lem:g_level1}
Let $w$ be the weight of the first level $w(0;i,j)$ in the symmetric lecture hall graph. Let 
\begin{equation*}
   g_{n,k}^w = \sum_{p:\SE((k,1)\to (n,0))}w(p). 
\end{equation*}
We have
\begin{equation*}
    g_{n,k}^w = w(0;k,k)g_{n,k+1}^w + g_{n-1,k-1}^{w^{-}},
\end{equation*}
where $w^{-} = w(0;i+1,j)$ for $j\leq i$ and $i,j\in\mathbf{Z}_{\geq 0}$. 
\end{lem}

\begin{lem}[\cite{Corteel2023}]
\label{lem:h_level2}
Let $w$ be the weight of the second level $w(1;i,j)$ in the symmetric lecture hall graph. Let 
\begin{equation*}
   h_{n,k}^w = \sum_{p:\SE((k,2)\to (n,1))}w(p). 
\end{equation*}
We have
\begin{equation*}
    h_{n,k}^w = w(0;n-1,0)h_{n-1,k}^w + h_{n-1,k-1}^{w^{+}},
\end{equation*}
where $w^{+} = w(0;i+1,j+1)$, $i,j\in\mathbf{Z}_{\geq 0}$. 
\end{lem}
\subsection{Connection with Motzkin paths}
When $b_n=0$ in \eqref{eq:3-term}, the mixed moment $\sigma_{n,k}$ equals zero whenever $n$ and $k$ have different parity. Therefore we consider the even, $\sigma_{2n,2k}$, and odd, $\sigma_{2n-1,2k-1}$, mixed moments separately, and establish the connection between Motzkin and symmetric lecture hall paths. When $b_n\neq 0$, the connection between Motzkin paths and lecture paths remains open. However, in the case when $b_n\neq 0$, one can colour the horizontal step in the Motzkin path and map it to either an appropriately weighted up-down or down-up step. This is a well-known bijection between 2-coloured Motzkin paths and Dyck paths \cite{ViennotLN}, see \Cref{thm:mot_slgh}. We illustrate this in the case of Laguerre polynomials in the Hydrogen atom section.

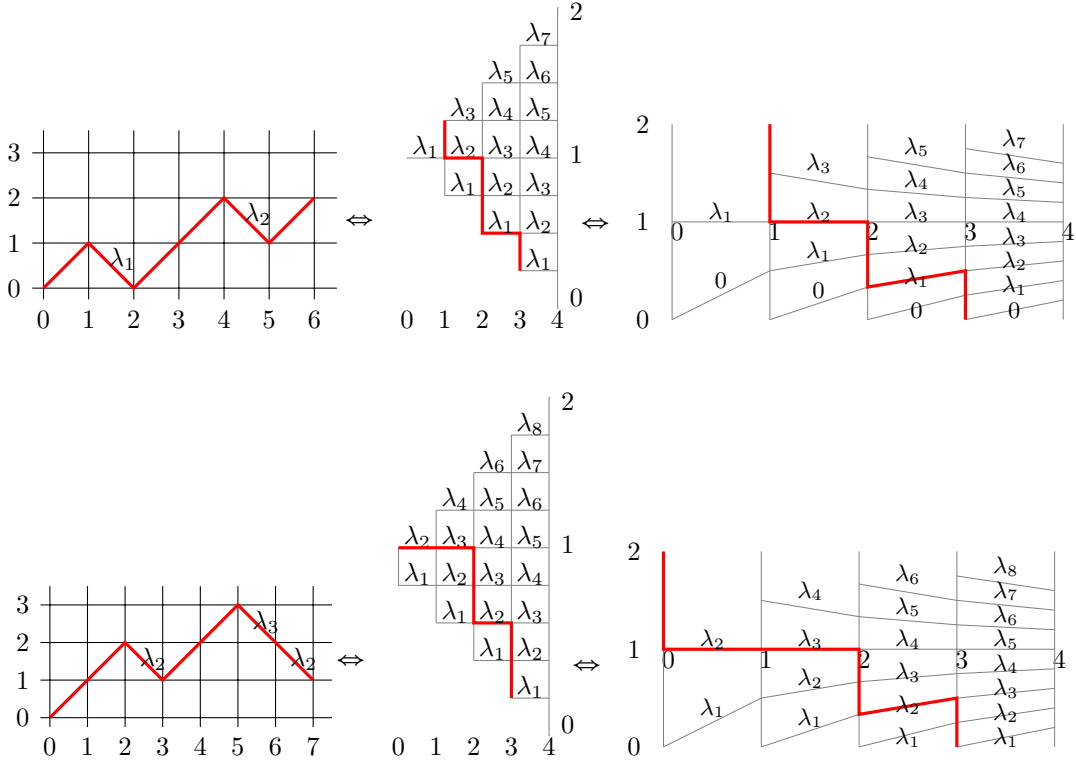
\begin{figure}
\begin{subfigure}[b]{\textwidth}
\centering
\begin{tikzpicture}[scale = 0.6]
\draw(-0.25,-0.25) grid (6.5,3.5);
\foreach \x in {0,1,...,6} { \node [anchor=north] at (\x,-0.3) {\x}; }
\foreach \y in {0,1,...,3} { \node [anchor=east] at (-0.3,\y) {\y}; }
\node at (7,1.5) {\(\Leftrightarrow\)};
\draw [red, very thick] (0,0) -- (1,1) -- (2,0) -- (3,1) -- (4,2) -- (5,1) -- (6,2);
\node [anchor=north] at (1.75,1.1) {$\lambda_1$};
\node [anchor=north] at (4.75,2.1) {$\lambda_2$};
\end{tikzpicture}
\begin{tikzpicture}[scale=0.5]
    \TRIL{4}
    \begin{scope} [shift={(0.5,0.3)}]
    \node at (0,0) {$\lambda_1$};
    \node at (1,0) {$\lambda_2$};
    \node at (1,1) {$\lambda_3$};
    \node at (2,0) {$\lambda_3$};
    \node at (2,1) {$\lambda_4$};
    \node at (2,2) {$\lambda_5$};
    \node at (3,0) {$\lambda_4$};
    \node at (3,1) {$\lambda_5$};
    \node at (3,2) {$\lambda_6$};
    \node at (3,3) {$\lambda_7$};
    \node at (4,-0.2) {$1$};
    \node at (4,3.6) {$2$};
    \end{scope}
    \begin{scope}[shift={(0,-4)}]
    \TRIDL{4} 
     \pgfmathsetmacro{\m}{4};
    \foreach \i in {0,...,\m}
    {
    \node at (\i,-.3) {\i};
    }
    \draw [red, very thick] (3,1) -- (3,2) -- (2,2) -- (2,3) -- (2,4) -- (1,4) -- (1,5);
\begin{scope} [shift={(0.5,0.3)}]
\node at (1,3) {$\lambda_1$};  
\node at (2,2) {$\lambda_1$};
\node at (2,3) {$\lambda_2$};
\node at (3,1) {$\lambda_1$};
\node at (3,2) {$\lambda_2$};
\node at (3,3) {$\lambda_3$};
\node at (4,0) {$0$};
\end{scope}
    \end{scope}
    \node at (5,-1.75) {\(\Leftrightarrow\)};
\end{tikzpicture}
\begin{tikzpicture}[scale=1.3]
\SLHlabel{4}
\SLH{4}
\draw [red, very thick] (3,0) -- (3,2/4) -- (2,1/3) -- (2,1) -- (1,1) -- (1,2);
\begin{scope}[shift={(0.5,0.1)}]
      \small
      \node at (0,0.3) {\( 0 \)};
      \node at (1,0.2) {\( 0\)};
      \node at (1,0.6) {\( \lambda_1\)};
      \node at (2,0) {\( 0 \)};
      \node at (2,1/3) {\( \lambda_1 \)};
      \node at (2,2/3) {\( \lambda_2\)};
      \node at (3,0) {\(0 \)};
      \node at (3,1/4) {\( \lambda_1\)};
      \node at (3,2/4) {\( \lambda_2\)};
      \node at (3,3/4) {\( \lambda_3 \)};
      \node at (0,1) {\( \lambda_1 \)};
      \node at (1,1) {\( \lambda_2 \)};
      \node at (1,3/2) {\( \lambda_3 \)};
      \node at (2,1) {\( \lambda_3 \)};
      \node at (2,4/3) {\( \lambda_4\)};
      \node at (2,5/3) {\( \lambda_5 \)};
      \node at (3,1) {\( \lambda_4 \)};
      \node at (3,5/4) {\( \lambda_5 \)};
      \node at (3,6/4) {\( \lambda_6 \)};
      \node at (3,7/4) {\( \lambda_7 \)};
    \end{scope}
\end{tikzpicture}
\end{subfigure}

\par\bigskip\bigskip

\begin{subfigure}[b]{\textwidth}
\centering
\begin{tikzpicture}[scale = 0.5]
\draw(-0.25,-0.25) grid (7.5,3.5);
\foreach \x in {0,1,...,7} { \node [anchor=north] at (\x,-0.3) {\x}; }
\foreach \y in {0,1,...,3} { \node [anchor=east] at (-0.3,\y) {\y}; }
\node at (8,1.5) {\(\Leftrightarrow\)};
\draw [red, very thick] (0,0) -- (2,2) -- (3,1) -- (5,3) -- (7,1);
\node [anchor=north] at (5.75,3.1) {$\lambda_3$};
\node [anchor=north] at (2.75,2.1) {$\lambda_2$};
\node [anchor=north] at (6.75,2.1) {$\lambda_2$};
\end{tikzpicture}
\begin{tikzpicture}[scale=0.5]
    \TRIL{4}
    \begin{scope} [shift={(0.5,0.3)}]
    \node at (0,0) {$\lambda_2$};
    \node at (1,0) {$\lambda_3$};
    \node at (1,1) {$\lambda_4$};
    \node at (2,0) {$\lambda_4$};
    \node at (2,1) {$\lambda_5$};
    \node at (2,2) {$\lambda_6$};
    \node at (3,0) {$\lambda_5$};
    \node at (3,1) {$\lambda_6$};
    \node at (3,2) {$\lambda_7$};
    \node at (3,3) {$\lambda_8$};
    \node at (4,-0.2) {$1$};
    \node at (4,3.6) {$2$};
    \end{scope}
    \begin{scope}[shift={(-1,-5)}]
    \TRIDL{5} 
     \pgfmathsetmacro{\m}{4};
    \foreach \i in {0,...,\m}
    {
    \node at (\i+1,-.3) {\i};
    }
    \draw [red, very thick] (4,1) -- (4,3) -- (3,3) -- (3,5) -- (1,5);
\begin{scope} [shift={(0.5,0.3)}]
\node at (1,4) {$\lambda_1$};  
\node at (2,3) {$\lambda_1$};
\node at (2,4) {$\lambda_2$};
\node at (3,2) {$\lambda_1$};
\node at (3,3) {$\lambda_2$};
\node at (3,4) {$\lambda_3$};
\node at (4,1) {$\lambda_1$};
\node at (4,2) {$\lambda_2$};
\node at (4,3) {$\lambda_3$};
\node at (4,4) {$\lambda_4$};
\node at (5,0) {$0$};
\end{scope}
    \end{scope}
    \node at (5,-3.15) {\(\Leftrightarrow\)};
\end{tikzpicture}
\begin{tikzpicture}[scale=1.3]
\SLHlabel{4}
\SLH{4}
\draw [red, very thick] (3,0) -- (3,2/4) -- (2,1/3) -- (2,1) -- (0,1) -- (0,2);
\begin{scope}[shift={(0.5,0.1)}]
      \small
      \node at (0,0.3) {\( \lambda_1 \)};
      \node at (1,0.2) {\( \lambda_1\)};
      \node at (1,0.6) {\( \lambda_2\)};
      \node at (2,0) {\( \lambda_1 \)};
      \node at (2,1/3) {\( \lambda_2 \)};
      \node at (2,2/3) {\( \lambda_3\)};
      \node at (3,0) {\(\lambda_1 \)};
      \node at (3,1/4) {\( \lambda_2\)};
      \node at (3,2/4) {\( \lambda_3\)};
      \node at (3,3/4) {\( \lambda_4 \)};
      \node at (0,1) {\( \lambda_2 \)};
      \node at (1,1) {\( \lambda_3 \)};
      \node at (1,3/2) {\( \lambda_4 \)};
      \node at (2,1) {\( \lambda_4 \)};
      \node at (2,4/3) {\( \lambda_5\)};
      \node at (2,5/3) {\(\lambda_6\)};
      \node at (3,1) {\( \lambda_5 \)};
      \node at (3,5/4) {\( \lambda_6 \)};
      \node at (3,6/4) {\( \lambda_7 \)};
      \node at (3,7/4) {\( \lambda_8\)};
    \end{scope}
\end{tikzpicture}
\end{subfigure}
\caption{Bijection among a Motzkin path, a path in staircase grid, and a path in the symmetric lecture hall graph. The weights of each east step are indicated in the staircase grid and the symmetric lecture hall graph.}
\label{fig:bijection}
\end{figure}

\begin{thm}
\label{thm:dyck_slgh}
In the three-term recurrence relation, let $b_n=0$ for $n\geq 0$, $\lambda_0=0$,  and $\lambda_n\neq 0$ for $n\geq 1$. Choose the weights $w_e$ and $w_o$ of the symmetric lecture hall graph to be
\begin{equation}
\label{eq:we_wo_gen}
    w_e(t;i,j) = \begin{cases}
        \lambda_j, & t=0,\\
        \lambda_{i+j+1}, & t=1,
    \end{cases}
    \qquad
    w_o(t;i,j) = \begin{cases}
      \lambda_{j+1}, & t=0,\\
      \lambda_{i+j+2}, & t=1.
    \end{cases}
\end{equation}
Then,
\begin{align*}
    \sigma_{2n,2k} &= \sum_{p:(0,0)\to (2n,2k)}w_{\Mot}(p) = \sum_{p:(k,2)\to (n,0)} w_e(p),\\
    \sigma_{2n+1,2k+1} &= \sum_{p:(0,0)\to (2n+1,2k+1)}w_{\Mot}(p) = \sum_{p:(k,2)\to (n,0)} w_o(p).
\end{align*}
\end{thm}
\begin{proof}
  First consider the case $\sigma_{2n,2k}$. Reflect the Motzkin path along the vertical axis and rotate by $45^{\circ}$ in the clockwise direction to obtain a new path $p$ that starts at $x=k$ and ends at $x=n$ with south and east steps, see the first bijection in \Cref{fig:bijection}. Denote the weight of the rotated path to be $\tilde{w}(p)$ (note that $\tilde{w}(p) = w_{\Mot}(p)$). The Motzkin path is translated to a path $p$ in the symmetric staircase grid such that it starts at $(k,2)$ and ends at $(n,0)$. 
  Now, 
\begin{align*}
    \sigma_{2n,2k} &= \sum_{p:(0,0)\to (2n,2k)}w_{\Mot}(p) \\
    & = \sum_{p:(k,2)\to (n,0)}\tilde{w}(p)\\
    & = \sum_{r=k}^n\left(\sum_{p:(k,2)\to (r,1)}\tilde{w}(p)\right)\left(\sum_{p:(r,1)\to (n,0)}\tilde{w}(p)\right)\\
    &= \sum_{r=k}^n\left(\sum_{p:(k,2)\to (r,1)}w_e(p)\right)\left(\sum_{p:(r,1)\to (n,0)}w_e(p)\right)\\
    & = \sum_{p:(k,2)\to (n,0)}w_e(p),
\end{align*} 
  where in the fourth line, we used the bijection between the lecture hall graph and its dual with the staircase grid, see the second bijection in \Cref{fig:bijection}. The case $\sigma_{2n+1,2k+1}$ can be proved in the same way.
\end{proof}

\begin{prop}
\label{prop:gh_eo}
    Consider weight systems $w_e$ and $w_o$ as given in \Cref{thm:dyck_slgh}. Also, consider $g^w_{n,k}$ and $h^w_{n,k}$ as denoted in \Cref{lem:g_level1} and $\Cref{lem:h_level2}$, respectively, when $w$ is either $w_e$ or $w_o$. Then,
\begin{align*}
    g^{w_e}_{n,k} &= \sigma_{2n-k-1,k-1}, \quad h^{w_e}_{n,k} = \sigma_{n,n,2k},\\
    g^{w_0}_{n,k} &= \sigma_{2n-k,k}, \quad h^{w_0}_{n,k} = \sigma_{n+1,n+1,2k+1}.
\end{align*}
\end{prop}
\begin{proof}
The weight of a SE step in Motzkin paths depends only on height $k\in\mathbf{Z_{\geq 0}}$ above the $x$-axis. Therefore, the weight of the path between the two levels $k$ and $k'$ remains the same if the path is shifted to the right arbitrarily by $x$ units. We can use the translational invariance of the Motzkin path to prove the proposition. 

First, consider the weight system $w_e$ in the symmetric lecture hall graph. Denote the weight of the symmetric staircase grid in the even and odd cases to be $\tilde{w}_e$ and $\tilde{w}_o$. For $k\geq 0$, the weight $w_e(p)$ of a path $p$ from $(2k+1,1)$ to $(n,0)$ is same as the weight $\tilde{w}_e(p')$ where $p'$ is a path from $(k,2)$ to $(n-k-1,0)$. Therefore,
\begin{equation}
\label{eq:g_sigma1}
    g^{w_e}_{n,2k+1} = \sum_{p:(2k+1,1)\to (n,0)} w_e(p)=\sum_{p':(k,2)\to (n-k-1,0)} \tilde{w}_e(p')= \sigma_{2n-2k-2,2k}.
\end{equation}
Similarly, for $k\geq 1$, the weight $w_e(p)$ of a path $p$ from $(2k,1)$ to $(n,0)$ is same as the weight $\tilde{w}_0(p')$ where $p'$ is a path from $(k-1,2)$ to $(n-k-1,0)$. Therefore,
\begin{equation}
\label{eq:g_sigma2}
    g^{w_e}_{n,2k}=\sum_{p:(2k,1)\to (n,0)} w_e(p)=\sum_{p':(k-1,2)\to (n-k-1,0)} \tilde{w}_o(p')= \sigma_{2n-2k-1,2k-1}.
\end{equation}
By combining \Cref{eq:g_sigma1,eq:g_sigma2}, for $k\geq 1$ we obtain
\begin{equation*}
    g^{w_e}_{n,k} =\sigma_{2n-k-1,k-1}.
\end{equation*}

To compute $h_{n,k}^{w_e}$ note that a path $p$ from $ (k,2)$ to $(n,1)$ can be extended to $(n,0)$ with the last $n+1$ south steps. In the case of Motzkin paths, it translates to a path that has the first $n$  NE steps fixed. So the Motzkin path starts at level $n$ and ends at level $2k$ and the length of the path is $n$. Therefore, $h_{n,k}^{w_e}$ can be written as
\begin{equation}
    h_{n,k}^{w_e} = \sum_{p:n\leadsto 2k,\,  |p|=n} w_{\Mot}(p) = \sigma_{n,n,2k}.
\end{equation}
The results for the $w_o$ case can be computed in the same way. 
\end{proof}

\begin{cor}
\label{cor:rec}
    We have 
    \begin{align*}
        \sigma_{2n,2k} &= \sum_{r=k}^n \sigma_{2n-r-1,r-1} \sigma_{r,r,2k},\\
        \sigma_{2n+1,2k+1} & = \sum_{r=k}^n\sigma_{2n-r,r}\sigma_{r+1,r+1,2k+1}.
    \end{align*}
\end{cor}
\begin{proof}
First consider the $\sigma_{2n,2k}$ case. By \Cref{thm:dyck_slgh} and \Cref{prop:gh_eo},
\begin{align*}
    \sigma_{2n,2k} &= \sum_{p:(k,2)\to (n,0)} w_e(p)= \sum_{r=k}^n\left(\sum_{p:(k,2)\to (r,1)}w_e(p)\right)\left(\sum_{p:(r,1)\to (n,0)}w_e(p)\right)= \sum_{r=k}^n \sigma_{2n-r-1,r-1}\sigma_{r,r,2k}.
\end{align*}
The recurrence relation can also be interpreted as follows. A Motzkin path from $(0,0)$ to $(2n,2k)$ can be divided into two paths, one from $(0,0)$ to $(2n-r-1,r-1)$ and the other from $(2n-r,r)$ to $(2n,2k)$ where a NE step connects the two paths. Using the translational invariance, the weight remains the same if the second part of the path is shifted left by $2n-2r$ units. Now, summing over all possible values of $r$ results in the recurrence relation i.e.
\begin{align*}
    \sigma_{2n,2k} & = \sum_{p:(0,0)\to (2n,2k)}w_{\Mot}(p)\\
    &= \sum_{r=k}^n\left(\sum_{p:(0,0)\to (2n-r-1,r-1)}w_{\Mot}(p)\right)\left(\sum_{p:(2n-r,r)\to (2n,2k)}w_{\Mot}(p)\right)\\
    &= \sum_{r=k}^n\left(\sum_{p:(0,0)\to (2n-r-1,r-1)}w_{\Mot}(p)\right)\left(\sum_{p:(r,r)\to (2r,2k)}w_{\Mot}(p)\right)\\
    &=\sum_{r=k}^n \sigma_{2n-r-1,r-1} \sigma_{r,r,2k}.
\end{align*}
The recurrence relation for $\sigma_{2n+1,2k+1}$ can be proved in the same way. 
\end{proof}

When $b_n\neq 0$ in \eqref{eq:3-term}, we have the following weight preserving bijection where Motzkin paths with all three steps, namely NE, E, and SE, are mapped to paths with only up and down steps. The NE step is mapped to a sequence of two up steps, the east step at level $k$ with weight $b_k$ is mapped to either a sequence of up and down steps at level $2k$ with $\gamma_{2k+1}$ as the weight of the down step or a sequence of down step at level $2k$ with weight $\gamma_{2k}$ and an up step.  The SE step from level $k$ to $k-1$ with weight $\lambda_k$ is mapped to two consecutive down steps from level $2k$ to $2k-2$ with weights $\gamma_{2k}$ and $\gamma_{2k-1}$, respectively. We choose the weights such that $b_k=\gamma_{2k+1}+\gamma_{2k}$ and $\lambda_k=\gamma_{2k}\gamma_{2k-1}$. See \Cref{fig:bijection_bn} for an example.  

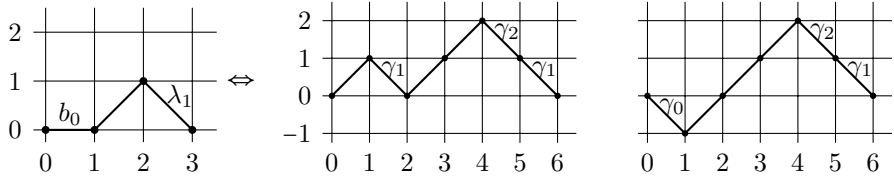
\begin{figure}[h]
\centering
\begin{tikzpicture}[scale = 0.65]
\draw(-0.25,-0.25) grid (3.5,2.5);
\foreach \x in {0,1,...,3} { \node [anchor=north] at (\x,-0.3) {\x}; }
\foreach \y in {0,1,2} { \node [anchor=east] at (-0.3,\y) {\y}; }
\draw [thick] (0,0) -- (1,0) -- (2,1)-- (3,0);
\filldraw[black](0,0)circle[radius=2pt];
\filldraw[black](1,0)circle[radius=2pt];
\filldraw[black](2,1)circle[radius=2pt];
\filldraw[black](3,0)circle[radius=2pt];
\node [anchor=north] at (0.5,0.75) {$b_0$};
\node [anchor=north] at (2.75,1.12) {$\lambda_1$};
 \node at (4,1) {\(\Leftrightarrow\)};
\end{tikzpicture}
\begin{tikzpicture}[scale = 0.5]
\draw(-0.25,-0.25) grid (6.5,3.5);
\foreach \x in {0,1,...,6} { \node [anchor=north] at (\x,-0.3) {\x}; }
\foreach \y in {0,...,3} { \node [anchor=east] at (-0.3,\y) {$\the\numexpr \y - 1\relax$}; }
\draw [thick] (0,1) -- (1,2) -- (2,1) -- (4,3) -- (6,1);
\filldraw[black](0,1)circle[radius=2pt];
\filldraw[black](1,2)circle[radius=2pt];
\filldraw[black](2,1)circle[radius=2pt];
\filldraw[black](3,2)circle[radius=2pt];
\filldraw[black](4,3)circle[radius=2pt];
\filldraw[black](5,2)circle[radius=2pt];
\filldraw[black](6,1)circle[radius=2pt];
\node [anchor=east] at (2.25,1.7) {$\gamma_1$};
\node [anchor=east] at (5.25,2.7) {$\gamma_2$};
\node [anchor=east] at (6.25,1.7) {$\gamma_1$};
\end{tikzpicture}
\hspace{0.5cm}
\begin{tikzpicture}[scale = 0.5]
\draw(-0.25,-0.25) grid (6.5,3.5);
\foreach \x in {0,1,...,6} { \node [anchor=north] at (\x,-0.3) {\x}; }
\draw [thick] (0,1) -- (1,0) -- (2,1) -- (4,3) -- (6,1);
\filldraw[black](0,1)circle[radius=2pt];
\filldraw[black](1,0)circle[radius=2pt];
\filldraw[black](2,1)circle[radius=2pt];
\filldraw[black](3,2)circle[radius=2pt];
\filldraw[black](4,3)circle[radius=2pt];
\filldraw[black](5,2)circle[radius=2pt];
\filldraw[black](6,1)circle[radius=2pt];
\node [anchor=east] at (1.25,0.7) {$\gamma_0$};
\node [anchor=east] at (5.25,2.7) {$\gamma_2$};
\node [anchor=east] at (6.25,1.7) {$\gamma_1$};
\end{tikzpicture}
\caption{A Motzkin path on the left with east, up and down steps is mapped to Motzkin paths with only up and down steps. The weight of each step is indicated.}
\label{fig:bijection_bn}
\end{figure}              

\begin{thm}
\label{thm:mot_slgh}
In the three-term recurrence relation, let $b_n\neq 0$ for $n\geq 0$, $\lambda_0=0$,  and $\lambda_n\neq 0$ for $n\geq 1$. By expressing $b_k=\gamma_{2k}+\gamma_{2k+1}$ and $\lambda_k=\gamma_{2k}\gamma_{2k-1}$, choose the weights $w$ of the symmetric lecture hall graph as
\begin{equation}
\label{eq:w_gen}
    w(t;i,j) = \begin{cases}
        \gamma_j, & t=0,\\
        \gamma_{i+j+1}, & t=1.
    \end{cases}
\end{equation}
Then,
\begin{align*}
    \sigma_{n,k} = \sum_{p:(k,2)\to (n,0)} w(p).
\end{align*}
\end{thm}
\begin{proof}
Using the bijection that maps Motzkin paths with three kinds of steps to paths with only up and down steps, and then using the bijection as shown in \Cref{fig:dlhg_grid} maps a Motzkin path to a path in the symmetric lecture hall graph. Then the theorem can be proved in the same way as \Cref{thm:dyck_slgh}.
\end{proof}

\begin{figure}[h]
\begin{tikzpicture}[scale=1.3]
\SLHlabel{4}
\SLH{4}
\begin{scope}[shift={(0.5,0.1)}]
      \small
      \node at (0,0.3) {\( \gamma_0 \)};
      \node at (1,0.2) {\( \gamma_0\)};
      \node at (1,0.6) {\( \gamma_1\)};
      \node at (2,0) {\( \gamma_0 \)};
      \node at (2,1/3) {\( \gamma_1 \)};
      \node at (2,2/3) {\( \gamma_2\)};
      \node at (3,0) {\(\gamma_0 \)};
      \node at (3,1/4) {\( \gamma_1\)};
      \node at (3,2/4) {\( \gamma_2\)};
      \node at (3,3/4) {\( \gamma_3 \)};
      \node at (0,1) {\( \gamma_1 \)};
      \node at (1,1) {\( \gamma_2 \)};
      \node at (1,3/2) {\( \gamma_3 \)};
      \node at (2,1) {\( \gamma_3 \)};
      \node at (2,4/3) {\( \gamma_4\)};
      \node at (2,5/3) {\( \gamma_5 \)};
      \node at (3,1) {\( \gamma_4 \)};
      \node at (3,5/4) {\( \gamma_5 \)};
      \node at (3,6/4) {\(\gamma_6 \)};
      \node at (3,7/4) {\( \gamma_7 \)};
    \end{scope}
\end{tikzpicture}
\caption{The symmetric lecture hall graph corresponding to the weights in \eqref{eq:w_gen}.}
\label{fig:lag_slhg}
\end{figure}
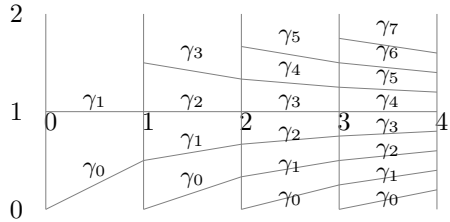

\section{Quantum Harmonic Oscillator}
\label{sec:qho}
Consider a quantum particle of mass $M$ confined to a one-dimensional region oscillating with an angular frequency $\omega$. The Hamiltonian of the particle  $\hat H$ is the sum of the kinetic and the potential energy,
\begin{equation*}
    \hat H = \frac{\hat{p}^2}{2M} + \frac{1}{2}M\omega^2 \hat{x}^2,
\end{equation*}
where $\hat{p}$ and $\hat{x}$ are the momentum and the position operators, respectively. The wavefunctions that solve the time-independent Schr\"{o}dinger equation in the position basis are the Hermite functions, given by
\begin{equation}
    \psi_m(x) = \frac{1}{\sqrt{m!}}\left(\frac{M\omega}{\pi\hbar}\right)^{\frac{1}{4}}e^{-\frac{M\omega}{2\hbar}x^2}H_m\left(\sqrt{\frac{2M\omega}{\hbar}}x\right), \quad m\geq 0.
\end{equation}
Here $H_m(x)$ are monic Hermite polynomials of degree $m$ that satisfy the three-term recurrence relation
\begin{equation*}
H_{m+1}(x) = x H_m(x) -  m H_{m-1}(x).
\end{equation*}
If we choose $M$ and $\omega$ such that $M=1,\,\, \omega=1$ and set $\hbar/2 =1$, the moments of the position operator are given by
\begin{equation}
\label{eq:pos_mom}
    \mathbb{E}[\hat x^n]\vert_{\psi_m} := \int_{-\infty}^\infty x^n |\psi_m(x)|^2 dx = \frac{1}{\sqrt{2\pi}m!}\int_{-\infty}^\infty x^n H_m^2(x) e^{-\frac{x^2}{2}} dx.
\end{equation}
As Hermite functions are the eigenfunctions of the Fourier operator, the moments of the momentum operator $\mathbb{E}[\hat p^n]\vert_{\psi_m}$ are also given by \eqref{eq:pos_mom}. Since the odd moments are zero, we focus on the even moments in \eqref{eq:pos_mom} and provide a combinatorial interpretation using lattice graphs. Applying Viennot's theory for orthogonal polynomials, the moments can be expressed as the generating function of Motzkin paths on the integer lattice. In addition, we give a new combinatorial interpretation in terms of lecture hall graphs.

\subsection{Moments of the position operator}
\subsubsection{Interpretation using Motzkin paths}
In this section, 
we give a quantum mechanical interpretation of Motzkin paths in terms of moments of the position operator of the harmonic oscillator. Note that since $b^H_n=0$ for Hermite polynomials, Motzkin paths consist of only up and down steps. Furthermore,
\begin{equation}
    \label{eq:single_step}
    \int_{-\infty}^\infty\psi_j(x)x\psi_k(x) dx = \frac{1}{\sqrt{2\pi j! k!}}\int_{-\infty}^\infty xH_j(x)H_k(x)e^{-\frac{x^2}{2}} dx = 
    \begin{cases}
        \sqrt{j+1}, & k=j+1,\\
        \sqrt{j}, & k=j-1,\\
        0, & \text{otherwise.}
    \end{cases}
\end{equation}
It is clear that if $k=j+1$, we take an up step from level $j$ to level $j+1$ whose weight is $\sqrt{\lambda^H_{j+1}}=\sqrt{j+1}$. Similarly, if $k=j-1$, we take a down step from level $j$ to level $j-1$ whose weight is $\sqrt{\lambda^H_{j}}=\sqrt{j}$. Unlike the Motzkin paths discussed in \Cref{sec:mot_paths}, we have modified weights where the weight of both the up and the down step between levels $j$ and $j+1$ is $\sqrt{j+1}$. 

\begin{prop}
\label{prop:two_weights}
For a Motzkin path $p$ that starts at an arbitrary level and ends at the same level, let $w$ be the weight system with the weight of the up step equal to 1, the weight of the horizontal step at level $j$ equal to $b_j$ and the weight of the down step from level $j$ to $j-1$ is $\lambda_j$. Consider a different weight system $\tilde{w}$ with the weight of the up and down steps between level $j-1$ and $j$ equal to $\sqrt{\lambda_j}$ and the weight of the horizontal step at level $j$ is $b_j$. Then
\begin{equation*}
    w_{\Mot}(p) = \tilde{w}_{\Mot}(p)
\end{equation*}
\end{prop}
\begin{proof}
Since the path $p$ always finishes at the same level as it began, the number of up steps and down steps are always equal. Therefore, the weight $w_{\Mot}(p)$ which is the product of the weights of all the steps in the path is the same as $\tilde{w}_{\Mot}(p)$.
\end{proof}

\begin{prop}
\label{prop:ho_mot}
    We have
    \begin{equation*}
        \mathbb{E}[\hat x^n]\vert_{\psi_m} = \sum_{p:m\leadsto m,\,  |p|=n} w^H_{\Mot}(p) =\sigma^H_{n,m,m}.
    \end{equation*}
\end{prop}
\begin{proof}
Using the identity
\begin{equation*}
\delta(x_1-x_2) = \sum_{j\geq 0}\psi_j(x_1)\psi_j(x_2),
\end{equation*}
the moments are given by
\begin{align*}
&\mathbb{E}[\hat x^n]\vert_{\psi_m}\\
=& \int_{-\infty}^\infty x^n \psi^2_m(x) dx\\
 =&  \int_{(-\infty,\infty)^n} \psi_m(x_1) x_1 \cdots x_n \psi_m(x_n)\prod_{j=1}^{n-1}\delta(x_j-x_{j+1}) dx_1\cdots dx_n\\
 =& \sum_{j_1,\cdots,j_{n-1}\geq 0} \int_{-\infty}^\infty\psi_m(x_1)x_1\psi_{j_1}(x_1)dx_1\left(\prod_{i=2}^{n-1}\int_{-\infty}^\infty\psi_{j_{i-1}}(x_i)x_i\psi_{j_i}(x_i) dx_i\right)\int_{-\infty}^\infty\psi_{j_{n-1}}(x_n)x_n\psi_{m}(x_n)dx_n\\
 =&\sum_{j_1,\cdots,j_{n-1}\geq 0}(\sqrt{m+1}\delta_{j_1,m+1}+\sqrt{m}\delta_{j_1,m-1})\prod_{i=2}^{n-1}(\sqrt{j_{i-1}+1}\delta_{j_{i},j_{i-1}+1}+\sqrt{j_{i-1}}\delta_{j_i,j_{i-1}-1})\\
 &\hspace{6.5cm}\times (\sqrt{m+1}\delta_{j_{n-1},m+1}+\sqrt{m}\delta_{j_{n-1},m-1}),
\end{align*}
which generates all Motzkin paths with up and down steps that start and end at level $m$, and the weight of each step between the levels $j$ and $j+1$ equal to $\sqrt{j+1}$. This proves the first equality. The second equality can be readily seen using \Cref{thm:lin_coef}. Since $\mathbb{E}[\hat x^n]\vert_{\psi_m}=\frac{1}{\sqrt{2\pi}m!}\int_{-\infty}^\infty x^n H_m^2(x) e^{-\frac{x^2}{2}} dx$, using \Cref{thm:lin_coef} it can also be shown directly that $\mathbb{E}[\hat x^n]\vert_{\psi_m} = \sigma^H_{n,m,m}$.
\end{proof}

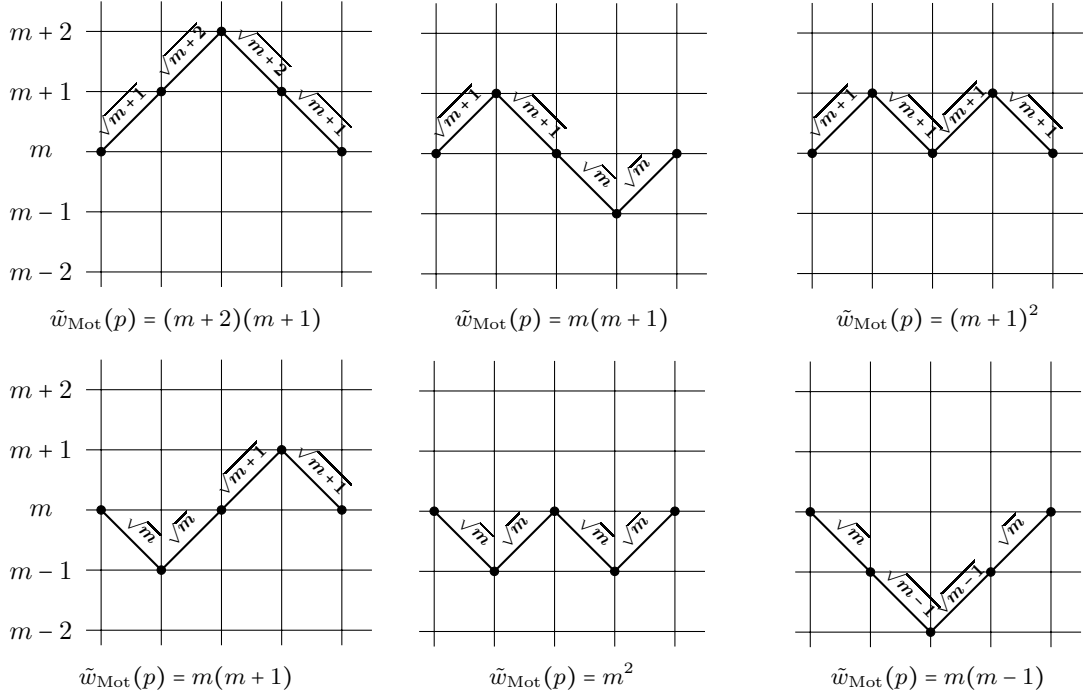
\begin{figure}[ht] 
    \begin{subfigure}[b]{0.33\linewidth}
    \centering
    \begin{tikzpicture}[scale = 0.8]
    \draw(-0.25,-0.25) grid (4.5,4.5);
    \node at (-1,4) {$m+2$};
    \node at (-1,3) {$m+1$};
    \node at (-1,2) {$m$};
    \node at (-1,1) {$m-1$};
    \node at (-1,0) {$m-2$};
    \draw [thick] (0,2) -- (1,3) -- (2,4) -- (3,3) -- (4,2);
    \filldraw[black](0,2)circle[radius=2pt];
    \filldraw[black](1,3)circle[radius=2pt];
    \filldraw[black](2,4)circle[radius=2pt];
    \filldraw[black](3,3)circle[radius=2pt];
    \filldraw[black](4,2)circle[radius=2pt];
    \node [anchor=south,rotate=45,scale=0.7] at (0.5,2.5) {$\bm{\sqrt{m+1}}$};
    \node [anchor=south,rotate=45,scale=0.7] at (1.5,3.5) {$\bm{\sqrt{m+2}}$};
    \node [anchor=south,rotate=-45,scale=0.7] at (2.5,3.5) {$\bm{\sqrt{m+2}}$};
    \node [anchor=south,rotate=-45,scale=0.7] at (3.5,2.5) {$\bm{\sqrt{m+1}}$};
    \end{tikzpicture}
    \captionsetup{labelformat=empty}
    \caption{$\tilde{w}_{\Mot}(p) =(m+2)(m+1)$}
    \end{subfigure}\hfill
    \begin{subfigure}[b]{0.33\linewidth}
    \centering
        \begin{tikzpicture}[scale = 0.8]
    \draw(-0.25,-0.25) grid (4.5,4.5);
    \draw [thick] (0,2) -- (1,3) -- (2,2) -- (3,1) -- (4,2);
    \filldraw[black](0,2)circle[radius=2pt];
    \filldraw[black](1,3)circle[radius=2pt];
    \filldraw[black](2,2)circle[radius=2pt];
    \filldraw[black](3,1)circle[radius=2pt];
    \filldraw[black](4,2)circle[radius=2pt];
    \node [anchor=south,rotate=45,scale=0.7] at (0.5,2.5) {$\bm{\sqrt{m+1}}$};
    \node [anchor=south,rotate=-45,scale=0.7] at (1.5,2.5) {$\bm{\sqrt{m+1}}$};
    \node [anchor=south,rotate=-45,scale=0.7] at (2.5,1.5) {$\bm{\sqrt{m}}$};
    \node [anchor=south,rotate=45,scale=0.7] at (3.5,1.5) {$\bm{\sqrt{m}}$};
    \end{tikzpicture}
   \captionsetup{labelformat=empty}
    \caption{$\tilde{w}_{\Mot}(p)=m(m+1)$}
    \end{subfigure}\hfill
    \begin{subfigure}[b]{0.33\linewidth}
    \centering
    \begin{tikzpicture}[scale = 0.8]
    \draw(-0.25,-0.25) grid (4.5,4.5);
    \draw [thick] (0,2) -- (1,3) -- (2,2) -- (3,3) -- (4,2);
    \filldraw[black](0,2)circle[radius=2pt];
    \filldraw[black](1,3)circle[radius=2pt];
    \filldraw[black](2,2)circle[radius=2pt];
    \filldraw[black](3,3)circle[radius=2pt];
    \filldraw[black](4,2)circle[radius=2pt];
    \node [anchor=south,rotate=45,scale=0.7] at (0.5,2.5) {$\bm{\sqrt{m+1}}$};
    \node [anchor=south,rotate=-45,scale=0.7] at (1.5,2.5) {$\bm{\sqrt{m+1}}$};
    \node [anchor=south,rotate=45,scale=0.7] at (2.65,2.55) {$\bm{\sqrt{m+1}}$};
    \node [anchor=south,rotate=-45,scale=0.7] at (3.5,2.5) {$\bm{\sqrt{m+1}}$};
    \end{tikzpicture}
    \captionsetup{labelformat=empty}
    \caption{$\tilde{w}_{\Mot}(p) =(m+1)^2$}
    \end{subfigure}

\bigskip

    \begin{subfigure}[b]{0.33\linewidth}
    \centering    \begin{tikzpicture}[scale = 0.8]
    \draw(-0.25,-0.25) grid (4.5,4.5);
    \node at (-1,4) {$m+2$};
    \node at (-1,3) {$m+1$};
    \node at (-1,2) {$m$};
    \node at (-1,1) {$m-1$};
    \node at (-1,0) {$m-2$};
    \draw [thick] (0,2) -- (1,1) -- (2,2) -- (3,3) -- (4,2);
    \filldraw[black](0,2)circle[radius=2pt];
    \filldraw[black](1,1)circle[radius=2pt];
    \filldraw[black](2,2)circle[radius=2pt];
    \filldraw[black](3,3)circle[radius=2pt];
    \filldraw[black](4,2)circle[radius=2pt];
    \node [anchor=south,rotate=-45,scale=0.7] at (0.5,1.5) {$\bm{\sqrt{m}}$};
    \node [anchor=south,rotate=45,scale=0.7] at (1.5,1.5) {$\bm{\sqrt{m}}$};
    \node [anchor=south,rotate=45,scale=0.7] at (2.5,2.5) {$\bm{\sqrt{m+1}}$};
    \node [anchor=south,rotate=-45,scale=0.7] at (3.5,2.5) {$\bm{\sqrt{m+1}}$};
    \end{tikzpicture}
    \captionsetup{labelformat=empty}
    \caption{$\tilde{w}_{\Mot}(p) =m(m+1)$}
    \end{subfigure}\hfill
    \begin{subfigure}[b]{0.33\linewidth}
    \centering
        \begin{tikzpicture}[scale = 0.8]
    \draw(-0.25,-0.25) grid (4.5,4.5);
    \draw [thick] (0,2) -- (1,1) -- (2,2) -- (3,1) -- (4,2);
    \filldraw[black](0,2)circle[radius=2pt];
    \filldraw[black](1,1)circle[radius=2pt];
    \filldraw[black](2,2)circle[radius=2pt];
    \filldraw[black](3,1)circle[radius=2pt];
    \filldraw[black](4,2)circle[radius=2pt];
    \node [anchor=south,rotate=-45,scale=0.7] at (0.5,1.5) {$\bm{\sqrt{m}}$};
    \node [anchor=south,rotate=45,scale=0.7] at (1.5,1.5) {$\bm{\sqrt{m}}$};
    \node [anchor=south,rotate=-45,scale=0.7] at (2.5,1.5) {$\bm{\sqrt{m}}$};
    \node [anchor=south,rotate=45,scale=0.7] at (3.5,1.5) {$\bm{\sqrt{m}}$};
    \end{tikzpicture}
   \captionsetup{labelformat=empty}
    \caption{$\tilde{w}_{\Mot}(p)=m^2$}
    \end{subfigure}\hfill
    \begin{subfigure}[b]{0.33\linewidth}
    \centering
    \begin{tikzpicture}[scale = 0.8]
    \draw(-0.25,-0.25) grid (4.5,4.5);
    \draw [thick] (0,2) -- (1,1) -- (2,0) -- (3,1) -- (4,2);
    \filldraw[black](0,2)circle[radius=2pt];
    \filldraw[black](1,1)circle[radius=2pt];
    \filldraw[black](2,0)circle[radius=2pt];
    \filldraw[black](3,1)circle[radius=2pt];
    \filldraw[black](4,2)circle[radius=2pt];
    \node [anchor=south,rotate=-45,scale=0.7] at (0.5,1.5) {$\bm{\sqrt{m}}$};
    \node [anchor=south,rotate=-45,scale=0.7] at (1.5,0.5) {$\bm{\sqrt{m-1}}$};
    \node [anchor=south,rotate=45,scale=0.7] at (2.65,0.55) {$\bm{\sqrt{m-1}}$};
    \node [anchor=south,rotate=45,scale=0.7] at (3.5,1.5) {$\bm{\sqrt{m}}$};
    \end{tikzpicture}
    \captionsetup{labelformat=empty}
    \caption{$\tilde{w}_{\Mot}(p) =m(m-1)$}
    \end{subfigure}
  \caption{ The fourth moment $\mathbb{E}[\hat x^4]\vert_{\psi_m}$ as generating function of Motzkin paths with modified weights. The weight of each up and down step is indicated. Each of these six paths corresponds to the expected value of a sequence of ladder operators, namely $\langle m|a^2{a^\dagger}^2|m\rangle$, $\langle m|a^\dagger a^2a^\dagger|m\rangle$, $\langle m|a a^\dagger a a^\dagger|m\rangle$, $\langle m| a {a^\dagger}^2a |m\rangle$, $\langle m|a^\dagger a a^\dagger a|m\rangle$, $\langle m|{a^\dagger}^2 a^2 |m\rangle$. The operators $a^\dagger$ and $a$ acting on $|m\rangle$ correspond to up and down steps at level $m$, respectively.}
  \label{fig:fourth_mom}
\end{figure}

\subsubsection{Ladder operators}
The Motzkin paths can also be interpreted in terms of raising/creation and lowering/annihilation operators of the quantum harmonic oscillator, namely $a^\dagger$ and $a$.  Denote the orthonormal energy eigenstates by $|n\rangle$, $n\in\mathbb{Z}_{\geq 0}$, so that $\hat H |n\rangle = E_n|n\rangle$, where the energy of the $n^{th}$ eigenstate is $E_n =2n+1$ in the units where we set $\hbar /2=\omega=1$. In the chosen units, the position and momentum operators can be expressed in terms of raising and lowering operators as
\begin{equation}
\label{eq:xp_aad}
\hat x = (a+a^\dagger), \quad \hat p = -i(a-a^\dagger).   
\end{equation}
The operators $a$ and $a^\dagger$ acts on the $n^{th}$ energy eigenstate as 
\begin{equation}
\label{eq:a_ad}
    a|n\rangle = \sqrt{n}|n-1\rangle, \quad a^\dagger |n\rangle = \sqrt{n+1}|n+1\rangle,
\end{equation}
and satisfies the commutation relation
\begin{equation}
\label{eq:com}
    [a,a^\dagger] =1.
\end{equation}

In terms of ladder operators, the moments of the position operator are given by
\begin{equation*}
    \mathbb{E}[\hat x^{2n}]\vert_{\psi_m}  = \langle m|(a+a^\dagger)^{2n}|m\rangle.
\end{equation*}
Using \eqref{eq:a_ad} and \eqref{eq:com}, it is clear that the number of $a$'s and $a^\dagger$'s in each term in the expansion of $(a+a^\dagger)^{2n}$ should be equal for the expectation to be non-zero. 
Thus, there are a total of $\frac{2n!}{n!n!}$ Motzkin paths that arise in the $2n^{th}$ moment of the position operator. For example, see \Cref{fig:fourth_mom} indicating all 6 paths that appear in $\mathbb{E}[\hat x^{4}]\vert_{\psi_m}$ in terms of ladder operators.

\subsubsection{Interpretation using lecture hall paths}
Since $b^H_n=0$, the even and odd mixed moments of Hermite polynomials, namely,
\begin{equation*}
   \sigma^H_{2n,2k} = \frac{2n!}{2^{n-k}2k! (n-k)!},\quad \sigma^H_{2n+1,2k+1} = \frac{(2n+1)!}{2^{n-k}(2k+1)! (n-k)!},
\end{equation*}
and the coefficients in the linearisation expansion can be expressed as generating functions of paths in the symmetric lecture hall graph as discussed in \Cref{sec:lhg}.

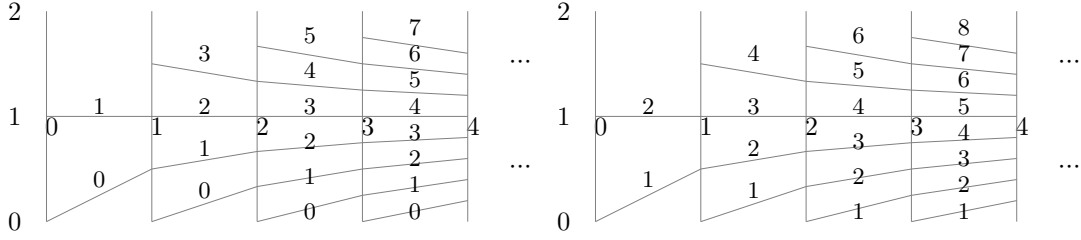
\begin{figure}
  \centering
  \begin{tikzpicture}[scale=1.4]
    \SLHlabel{4}
    \SLH{4}
    \begin{scope}[shift={(0.5,0.1)}]
      \small
      \node at (0,0.3) {\( 0 \)};
      \node at (1,0.2) {\( 0 \)};
      \node at (1,0.6) {\( 1 \)};
      \node at (2,0) {\( 0 \)};
      \node at (2,1/3) {\( 1 \)};
      \node at (2,2/3) {\( 2 \)};
      \node at (3,0) {\( 0 \)};
      \node at (3,1/4) {\(1 \)};
      \node at (3,2/4) {\(2\)};
      \node at (3,3/4) {\(3 \)};
      \node at (0,1) {\( 1 \)};
      \node at (1,1) {\( 2 \)};
      \node at (1,3/2) {\( 3 \)};
      \node at (2,1) {\( 3 \)};
      \node at (2,4/3) {\( 4 \)};
      \node at (2,5/3) {\( 5 \)};
      \node at (3,1) {\( 4 \)};
      \node at (3,5/4) {\(5 \)};
      \node at (3,6/4) {\( 6 \)};
      \node at (3,7/4) {\( 7 \)};
    \end{scope}
    \node at (4.5,0.5) {$\cdots$};
    \node at (4.5,1.5) {$\cdots$};
    \end{tikzpicture}
  \begin{tikzpicture}[scale=1.4]
    \SLHlabel{4}
    \SLH{4}
    \begin{scope}[shift={(0.5,0.1)}]
      \small
      \node at (0,0.3) {\( 1 \)};
      \node at (1,0.2) {\( 1 \)};
      \node at (1,0.6) {\( 2 \)};
      \node at (2,0) {\( 1 \)};
      \node at (2,1/3) {\( 2 \)};
      \node at (2,2/3) {\( 3 \)};
      \node at (3,0) {\( 1 \)};
      \node at (3,1/4) {\(2 \)};
      \node at (3,2/4) {\(3\)};
      \node at (3,3/4) {\(4 \)};
      \node at (0,1) {\( 2 \)};
      \node at (1,1) {\( 3 \)};
      \node at (1,3/2) {\( 4 \)};
      \node at (2,1) {\( 4 \)};
      \node at (2,4/3) {\( 5 \)};
      \node at (2,5/3) {\( 6 \)};
      \node at (3,1) {\( 5 \)};
      \node at (3,5/4) {\(6 \)};
      \node at (3,6/4) {\( 7 \)};
      \node at (3,7/4) {\( 8 \)};
    \end{scope}
    \node at (4.5,0.5) {$\cdots$};
    \node at (4.5,1.5) {$\cdots$};
    \end{tikzpicture}
    \caption{The weights for $\sigma^H_{2n,2k}$ (left) and $\sigma^H_{2n+1,2k+1}$ (right) for Hermite polynomials.}
    \label{fig:her_mixmom}
\end{figure}

\begin{prop}
    The mixed moments of Hermite polynomials can be written as
    \begin{equation*}
        \sigma^H_{2n,2k} = \sum_{p:(k,2)\to (n,0)}w^H_e(p), \quad \sigma^H_{2n+1,2k+1} = \sum_{p:(k,2)\to (n,0)}w^H_o(p),
    \end{equation*}
    where $w^H_e$ and $w^H_o$ are given in \eqref{eq:we_wo_gen} with $\lambda^H_0=0$ and $\lambda^H_j=j$.
\end{prop}
\begin{proof}
    The proof follows from \Cref{thm:dyck_slgh} using the bijections in \Cref{fig:bijection}. Alternatively, the results can also be seen directly by combining \Cref{prop:gh_eo} and \Cref{cor:rec}.
\end{proof}

Now, the moments of the position operator are given by
\begin{equation*}
    \mathbb{E}[\hat x^{2n}]\vert_{\psi_{2m}}  = \sigma^H_{2n,2m,2m} = \sum_{p:(m,2)\to (m+n,\frac{2m+1}{m+n+1})}w^H_e(p) = \sum_{p:(m,\frac{2m+1}{m+1})\to (m+n,\frac{2m}{m+n+1})}w^H_o(p).
\end{equation*}
Note that the last $m+1$ and $m$ south steps in $w_e^H$ and $w_o^H$, respectively, are fixed in both cases. One can also utilise the translational invariance of the weighted path to write
\begin{equation*}
    \mathbb{E}[\hat x^{2n}]\vert_{\psi_{2m}}  = \sigma_{2n,2m,2m} = \sum_{p:(m+k,\frac{2m+2}{m+k+1})\to (m+n+k,\frac{2m+1}{m+n+k+1})}w^H_e(p) = \sum_{p:(m+k,\frac{2m+1}{m+k+1})\to (m+n+k,\frac{2m}{m+n+k+1})}w^H_o(p),
\end{equation*}
where $k\geq 0$. 
Similarly,
\begin{equation*}
    \mathbb{E}[\hat x^{2n}]\vert_{\psi_{2m+1}}  = \sigma_{2n,2m+1,2m+1} = \sum_{p:(m+1,\frac{2m+3}{m+2})\to (m+n+1,\frac{2m+2}{m+n+2})}w^H_e(p) = \sum_{p:(m,2)\to (m+n,\frac{2m+1}{m+n+1})}w^H_o(p),
\end{equation*}
and for any $k\geq 0$,
\begin{equation*}
    \mathbb{E}[\hat x^{2n}]\vert_{\psi_{2m}}  = \sigma_{2n,2m+1,2m+1} = \sum_{p:(m+k+1,\frac{2m+3}{m+k+2})\to (m+n+k+1,\frac{2m+2}{m+n+k+2})}w^H_e(p) = \sum_{p:(m+k,\frac{2m+2}{m+k+1})\to (m+n+k,\frac{2m+1}{m+n+k+1})}w^H_o(p).
\end{equation*}

\section{Hydrogen atom}
\label{sec:hatom}
The total wavefunction of the hydrogen atom is a product of the radial wave function and the angular wave function. The radial wave function, denoted as $R(r)$, represents the probability density of finding the electron at a particular radius $r$ from the nucleus. The radial wavefunction with orbital angular momentum quantum number $l$ and energy level $n$ is given by
\begin{align}
    R_{nl}(r) &= (-1)^{n-l-1}\sqrt{\left(\frac{2}{na_0}\right)^3\frac{1}{2n(n+l)!(n-l-1)!}}\left(\frac{2r}{na_0}\right)^{l}e^{-\frac{r}{na_0}} L_{n-l-1}^{(2l+1)}\left(\frac{2r}{na_0}\right). 
\end{align}
Here $a_0$ is the Bohr radius and $L_m^{(\alpha)}(x)$ is the monic associated Laguerre polynomial of degree $m$ orthogonal to the weight $x^\alpha e^{-x}$.  Polynomials $L_m^{(\alpha)}(x)$ satisfy the recurrence relation
\begin{equation*}
    L_{m+1}^{(\alpha)}(x) = (x-b_m^L)L_{m}^{(\alpha)}(x) - \lambda_m^L L_{m-1}^{(\alpha)}(x),
\end{equation*}
where $b_m^L = 2m+\alpha+1$ and $\lambda_m^L = m(m+\alpha)$. The quantum numbers $n$ and $l$ are also referred to as the principal quantum number and the azimuthal quantum number, respectively. The statement of orthonormality for the radial wavefunction is 
\begin{equation*}
    \int_0^\infty R_{nl}(r)R_{n^\prime l}(r)r^2 dr = \delta_{n n^\prime}.
\end{equation*}

By relabelling $x=2r/na_0$, the probability density of the electron with principal quantum number $n$ and azimuthal quantum number $l$ is 
\begin{equation}
    p_{nl}(x) = \frac{1}{2n(n-l-1)!(n+l)!}e^{-x}x^{2l+2}\left(L_{n-l-1}^{(2l+1)}(x)\right)^2.
\end{equation}
Therefore, the moments of the position operator are given by
\begin{equation}
    \mathbb{E}[\hat x^n]\vert_{R_{ml}} = \frac{1}{2m(m-l-1)!(m+l)!}\int_0^\infty x^{n+2l+2} \left(L_{m-l-1}^{(2l+1)}(x)\right)^2 e^{-x}dx. 
\end{equation}

For Laguerre polynomials, we have
\begin{equation*}
    \frac{1}{\left(i!j!\Gamma(i+\alpha+1)\Gamma(j+\alpha+1)\right)^{1/2}}\int_0^\infty x L_i^{(\alpha)}(x) L_j^{(\alpha)}(x) x^\alpha e^{-x} dx = 
    \begin{cases}
    \sqrt{(i+1)(i+\alpha+1)}=\sqrt{\lambda^L_{i+1}}, & j=i+1,\\
    2i+\alpha+1 = b^L_i, & j=i,\\
    \sqrt{i(i+\alpha)} = \sqrt{\lambda^L_i}, &j=i-1,
    \end{cases}
\end{equation*}
which can be used to generate all modified Motzkin paths with the weight of the up and down steps between level $j$ and $j+1$ equal to $\sqrt{\lambda^L_j}$, and the weight of the east step at level $j$ equal to $b_j^L$. Therefore, similar to \Cref{prop:ho_mot}, we have
    \begin{equation*}
        \mathbb{E}[\hat x^n]\vert_{R_{ml}} = \sum_{p:m-l-1\leadsto m-l-1,\,  |p|=n+1} w^L_{\Mot}(p) =\sigma^L_{n+1,m-l-1,m-l-1},
    \end{equation*}
where $w^L_{\Mot}(p)$ is a weight of the Motzkin path, and $\sigma^L_{i,j,k}$ is the linearisation coefficient corresponding to the Laguerre polynomials. 

Since $b_k^L = 2k+\alpha+1$ and $\lambda_k^L=k(k+\alpha)$, using \Cref{thm:mot_slgh} the weight of a path $w^L(p)$ in the symmetric lecture hall graph can be computed by setting $\gamma^L_{2k}=k+\alpha$ and $\gamma^L_{2k+1}=k+1$, see \Cref{fig:lag_slhg}. Therefore, we have
\begin{equation}
    \sigma_{n,k}^L = \binom{n}{k}\frac{\Gamma(n+\alpha+1)}{\Gamma(k+\alpha+1)} = \sum_{p:(k,2)\to (n,0)}w^L(p).
\end{equation}

Similarly, the moments of the position operator of the hydrogen atom can be expressed as
\begin{equation}
    \mathbb{E}[\hat x^n]\vert_{R_{ml}} = \sum_{p:(m-l-1,2)\to (m+n-l,\frac{2m-2l-1}{m+n-l+1})}w^L(p).
\end{equation}

\section*{Acknowledgments}

The authors are grateful to Jang Soo Kim for extremely helpful discussions concerning the symmetric lecture hall graph.

\bibliographystyle{abbrv}
\bibliography{refs}

\end{document}